\documentclass[a4paper,twocolumn,11pt,accepted=2026-10-05]{quantumarticle}
\pdfoutput=1
\usepackage[utf8]{inputenc}
\usepackage[english]{babel}
\usepackage[T1]{fontenc}
\usepackage{hyperref}
\usepackage[numbers]{natbib}
\usepackage{tikz}
\usepackage{lipsum}

\usepackage{amsmath,braket}
\usepackage{amsfonts}
\usepackage{amssymb}
\usepackage{graphicx} 
\usepackage{xcolor}
\usepackage{amsthm}
\usepackage{enumitem}
\usepackage{mathtools}
\theoremstyle{plain}
\newtheorem{theorem}{Theorem}[section]
\newtheorem{lemma}[theorem]{Lemma}
\newtheorem{proposition}[theorem]{Proposition}

\newtheorem{conjecture}[theorem]{Conjecture}
\theoremstyle{definition}

\theoremstyle{remark}
\newtheorem{remark}[theorem]{Remark}
\renewcommand{\H}{\mathcal{H}}
\newcommand{\K}{\mathcal{K}}
\newcommand{\X}{\mathcal{X}}
\newcommand{\Y}{\mathcal{Y}}
\newcommand{\Z}{\mathcal{Z}}
\newcommand{\I}{\mathds{1}}
\usepackage{dsfont}
\newcommand{\id}{\operatorname{Id}}

\newcommand{\red}[1]{{#1}}

\begin{document}
\title{A Counterexample to the Optimality Conjecture in Convex Quantum Channel Optimization}

\author{Jianting Yang}
 \affiliation{School of Mathematical Sciences, University of Electronic Science and Technology of China}
\affiliation{State Key Laboratory of Mathematical Sciences, Academy of Mathematics and Systems Science, Chinese Academy of Sciences}
\orcid{0009-0000-2357-2114}
\email{yangjianting@uestc.edu.cn}
\maketitle
\begin{abstract}
 This paper presents a counterexample to the optimality conjecture for the  trace distance based optimal state transformation problem  proposed by Coutts et al. [Quantum 5, 448 (2021)]. The conjecture posits that for the  trace distance based optimal state transformation problem, \red{applying the sign function to the error
  matrix uniquely determines a dual certificate for every optimal solution}. By constructing a counterexample in 2-dimensional Hilbert spaces, we disprove this conjecture. 
\end{abstract}
\section{Introduction}
The optimization problems  over quantum channels and measurements are  fundamental to quantum information theory.  Previous studies have established optimality conditions for various quantum information optimization problems, ranging from the optimal quantum measurement problem~\cite{HOLEVO1973337,1055351} to  semidefinite programming formulations for quantum state discrimination~\cite{1193807}. Mathematically, many of these optimization problems can be formulated as convex optimization problems, especially those defined over the semidefinite cone that involve minimizing the nuclear norm (trace norm) of \red{an affine function of the  Choi matrix}. While numerical methods such as semidefinite programming (SDP) are available to solve these problems, verifying the optimality of the solutions remains a significant theoretical and computational challenge.

The certification of optimality is an important component in the field of optimization. For general convex optimization problems, the Karush-Kuhn-Tucker (KKT) conditions offer a general method to verify optimality, typically requiring the computation of dual problems. However, in some special cases, including group synchronization~\cite{MR4590243,MR3982681}, semidefinite programming~\cite{MR4914032}, and quantum channel optimization~\cite{Coutts2021certifying}, the certification of optimality without solving dual problems has become an effective approach. These methods use properties of primal solutions, avoiding computational costs for dual problems.

This paper considers the optimality conditions for the trace distance based  optimal state transformation problem. 
Coutts et al.  conjectured that \red{a} dual certificate could be uniquely determined through the spectral calculus of the \red{error matrix} of \red{an} optimal channel \cite[Conjecture~11]{Coutts2021certifying}.  They provided optimality conditions based on the subdifferential of the trace norm,  proving that when the error matrix of an optimal channel is full-rank, the dual certificate is uniquely determined via the spectral calculus of the \red{error matrix}. However, the general case where the error matrix is a singular (rank-deficient) matrix still remains unsolved. In this case, the subdifferential is no longer a singleton. Based on  numerical experiments, the authors conjectured that a  specific operator by setting the sign of zero eigenvalues to zero  would  provide  a necessary and sufficient condition for optimality. If this  conjecture is true, this would imply that the optimality of a quantum channel could always be certified by a  spectral calculus efficiently. By introducing additional symmetry to the trace distance optimization problem, we reduce the nuclear norm minimization problem over the  $4 \times 4$ positive semidefinite cone  to an equivalent 2-dimensional  linear programming (LP) problem.    This reduction allows us to construct a counterexample, where the conjectured spectral condition fails to provide a dual certificate.
    
A key theoretical approach in this work is the adoption  of symmetry to simplify the optimization problems, which is an important approach for solving large-scale optimization problems.  When  an optimization problem admits symmetric structures, it can be reduced to a simpler one  via group theory and operator algebra theory.  In the fields of polynomial optimization and matrix optimization, symmetry reduction methods have been extensively studied~\cite{MR2067190,MR3029481}.
\subsection{Contributions}
This work provides a negative answer to the conjecture by Coutts et al. Our contributions are summarized as follows:
 \paragraph{Symmetry Reduction:} We propose a symmetry reduction approach that transforms the optimization problem from the \red{high-dimensional convex set of quantum channels (i.e., completely positive and trace-preserving linear maps)} to a low-dimensional linear programming (LP) problem. This dimensional reduction substantially narrows the search space of counterexamples.
  \paragraph{Counterexample to the Conjecture:} An explicit counterexample is constructed to the conjecture proposed by Coutts et al.~\cite{Coutts2021certifying}. This construction shows that the conjectured spectral calculus of the \red{error matrix}  is insufficient to certify optimality.

	\section{Notations and Preliminaries}
	Let $\mathcal{H}$ and $\mathcal{K}$ be finite-dimensional complex  Hilbert spaces, $\mathcal{H} \otimes \mathcal{K}$ their tensor product, $L(\mathcal{H})$   the set of bounded linear operators acting on $\mathcal{H}$, $\mathsf{D}(\mathcal{H}) \subseteq L(\mathcal{H})$ the set of density operators on $\mathcal{H}$, \red{$\mathsf{CP}(\mathcal{H},\mathcal{K})$ the set of completely positive linear maps from $L(\mathcal{H})$ to $L(\mathcal{K})$, and $\mathsf{C}(\mathcal{H},\mathcal{K})\subseteq\mathsf{CP}(\mathcal{H},\mathcal{K})$ the set of completely positive and trace-preserving linear maps (quantum channels) from $L(\mathcal{H})$ to $L(\mathcal{K})$},  and $\I_{\mathcal{H}} \in L(\mathcal{H})$ is the identity operator. Moreover, $\id_{L(\mathcal{H})}: L(\mathcal{H}) \to L(\mathcal{H})$ denotes  the identity map on $L(\mathcal{H})$. Let  $\operatorname{Tr}_{\mathcal{H}}: L(\mathcal{H} \otimes \mathcal{K}) \to L(\mathcal{K})$ be the partial trace over $\mathcal{H}$, i.e., for any operator $X = \sum_i A_i \otimes B_i \in L(\mathcal{H} \otimes \mathcal{K})$ with $A_i \in L(\mathcal{H})$ and $B_i \in L(\mathcal{K})$ 
	\[\operatorname{Tr}_{\mathcal{H}}(X) = \sum_i \operatorname{Tr} (A_i) B_i,\]
	where $\operatorname{Tr}(\cdot)$ denotes the standard trace. For a \red{quantum channel} $\Phi \in \mathsf{C}(\H,\K)$, define its Choi representation $J(\Phi)\in L(\K\otimes\H)$  as
	\begin{equation}\label{equ:def-J}
		J(\Phi) \coloneq \sum_{i,k=0}^{d-1} \Phi\left(\ket{i}_\H\bra{k}\right) \otimes \ket{i}_\H\bra{k},
	\end{equation} 
	where $\{\ket{i}_{\H}\}_{i=0}^{d-1}$ is a given orthonormal basis of $\H$, and   $\{\ket{i}_\H\bra{k}\}_{i,k=0}^{d-1}$ is the canonical basis of the space of linear operators acting on $\H$.
	
	In \cite{Coutts2021certifying}, the authors have the following conjecture regarding the optimality condition of a quantum channel:
	\begin{conjecture}[{\cite[Conjecture~11]{Coutts2021certifying}}]\label{conj:main} Let $\X$, $\Y$, and $\Z$ be finite-dimensional complex Hilbert spaces, for given  density operators $\rho \in \mathsf{D}(\X\otimes \Z)$ and $\sigma \in \mathsf{D}(\Y\otimes \Z)$. Let $\Phi \in \mathsf{C}(\X, \Y)$ be a \red{quantum channel}, define $Y$ by the following spectral calculation:
		\[Y=\operatorname{sign}\left(\sigma-(\Phi\otimes \red{\id_{L(\Z)}} )(\rho)\right),\]
		then $\Phi$ is an optimal solution to the nuclear norm minimization problem  
		\begin{equation}\label{opt:org}
			\begin{aligned}
				\min_{\Phi} \quad & \| \sigma-(\Phi\otimes \red{\id_{L(\Z)}} )(\rho)\|_{*} \\
				\text{s.t.} \quad &\Phi \in \mathsf{C}(\X, \Y),
			\end{aligned}
		\end{equation}
		if and only if the operator 
        \begin{equation}\label{equ:def:H}
           H=\left(\id_{L(\Y)}\otimes \Psi_{\rho}^* \right)(Y) 
        \end{equation}
         satisfies  $\operatorname{Tr}_{\Y}(HJ(\Phi)) \in \operatorname{Herm}(\X)$
        and 
		\begin{equation}\label{equ:Conj}
			\ H \succeq \I_{\Y}\otimes \operatorname{Tr}_{\Y}(HJ(\Phi)),
		\end{equation}
		where \red{$\Psi_{\rho}\in \mathsf{CP}(\X,\Z)$} is the completely positive map satisfying  the condition that for all $\Phi \in  \mathsf{C}(\X,\Y)$,
		\begin{equation}\label{equ:psi}
			(\Phi\otimes \id_{L(\Z)})(\rho)=(\id_{L(\Y)}\otimes \Psi_{\rho})(J(\Phi)),
		\end{equation}
        and $\Psi_{\rho}^*$ is its adjoint operator.
	\end{conjecture}

    Moreover, in \cite[Corollary~10]{Coutts2021certifying}, the authors proved that when $\sigma-(\Phi\otimes \red{\id_{L(\Z)}} )(\rho)$ is a full-rank matrix,  the conjecture holds true. Therefore, this conjecture is a natural extension of their result to the general case.

	\section{Reduction via Symmetry Structure}\label{sec:2}
	Problem~\eqref{opt:org} is defined over the \red{convex set of quantum channels},  making it challenging to directly construct a counterexample. In this section, we introduce several specific symmetry structures to the Problem~\eqref{opt:org},  which allow  us to reduce the original  optimization problem from a high-dimensional  positive semidefinite cone to a lower-dimensional linear programming (LP) problem.

	Let $\X$, $\Y$, and $\Z$ be 2-dimensional complex Hilbert spaces. To \red{differentiate between} their bases, these bases are explicitly distinguished by denoting the basis of $\X$ as    $\{ \ket{0}_{\X},\ket{1}_{\X}\}$, and similarly for  $\Y$ and $\Z$. 

    \subsection{Reduction via the Symmetry of the Density Operator}
	\red{Define the density operator corresponding to the Bell state} $\rho_0 \in \mathsf{D}(\X \otimes \Z)$  with
	\begin{equation}\label{equ:def-rho}
		\rho_0 \coloneq \frac{1}{2}\sum_{i=0}^1\sum_{j=0}^1 \ket{i}_{\X}\bra{j}\otimes \ket{i}_{\Z}\bra{j}.
	\end{equation} 
	Then we have the following result for such $\rho_0$.
	\begin{proposition}\label{prop:rho_0}
	 Let $\rho_0$ be defined as in \eqref{equ:def-rho}   and $\Psi_{\rho_0}$ be the associated completely positive map. Then  for any $\Phi \in \mathsf{C}(\X, \Y)$,  
	 	\begin{equation}\label{equ:Id-tensor-Psi-is-half}
	 	(\id_{L(\Y)}\otimes \Psi_{\rho_0})(J(\Phi))=\frac{1}{2}J(\Phi)
	 \end{equation}
	 under  the linear space isomorphism $\mathcal{I}_{\X\to\Z}: \X\to\Z$ defined by $\ket{i}_\X \mapsto \ket{i}_\Z$.
	\end{proposition}
	\begin{proof}
		For any $\Phi \in \mathsf{C}(\X, \Y)$, by \eqref{equ:def-J}, we have
		\begin{equation}\label{equ:JPhi}
			J(\Phi)=\sum_{i=0}^1 \sum_{j=0}^1 \Phi(\ket{i}_{\X}\bra{j}) \otimes\ket{i}_{\X}\bra{j}
		\end{equation}
		and by \eqref{equ:psi} we also have 
		\begin{equation}\label{equ:Id-tensor-Psi}
        \begin{aligned}
			(\Phi\otimes \id_{L(\Z)})(\rho_0)&=(\id_{L(\Y)}\otimes \Psi_{\rho_0})(J(\Phi)) \\
            &= \sum_{i,j=0}^1 \Phi(\ket{i}_{\X}\bra{j})  \otimes 
			 \frac{1}{2}\ket{i}_{\Z}\bra{j}.
         \end{aligned}
		\end{equation}
By comparing \eqref{equ:JPhi} and \eqref{equ:Id-tensor-Psi}, one can obtain 
		\begin{equation*}
			(\id_{L(\Y)}\otimes \Psi_{\rho_0})(J(\Phi))=\frac{1}{2}J(\Phi)
		\end{equation*}
		under  the linear space isomorphism $\mathcal{I}_{\X\to\Z}$.
	\end{proof}

	Clearly, for $\rho_0$ defined in \eqref{equ:def-rho} and any $\sigma \in \mathsf{D}(\Y \otimes \Z)$, the nuclear norm optimization problem~\eqref{opt:org} can be reformulated as follows.
	\begin{equation*}
		\begin{aligned}
			\min_{\tilde{X}} \quad & \| \sigma -(\id_{L(\Y)}\otimes \Psi_{\rho_0})({\red{\tilde{X}}})\|_{*} \\
			\text{s.t.} \quad &\tilde{X} \in J(\mathsf{C}(\X, \Y)).
		\end{aligned}
	\end{equation*}

By Proposition~\ref{prop:rho_0}, the map $(\id_{L(\Y)} \otimes \Psi_{\rho_0})$ is a scalar multiplication under the isomorphism. Formally, let $\I_{\Y}\otimes \mathcal{I}_{\X\to\Z}$ be the linear isomorphism mapping $\Y\otimes\X$ to $\Y\otimes\Z$, with $\ket{i}_{\Y}\otimes \ket{j}_{\X} \to \ket{i}_{\Y}\otimes \ket{j}_{\Z}$. 

Since
\begin{align*}
J(\mathsf{C}(\X, \Y)) = \left\{ \tilde{X} \in L(\Y \otimes \X);\ \tilde{X}\succeq 0, \right. \notag
\\
\left. \operatorname{Tr}_{\Y}(\tilde{X}) = \I_{\X} \right\},
\end{align*}
the optimization problem~\eqref{opt:org} is equivalent to:
\begin{equation}\label{eq:opt_rigorous}
    \begin{aligned}
        \min_{\tilde{X}\in L(\Y \otimes \X)} \quad & \left\| \sigma - \frac{1}{2} \mathcal{U} \tilde{X}\mathcal{U}^\dag \right\|_* \\
        \text{s.t.} \quad &  \tilde{X}\succeq 0, \\
                          & \operatorname{Tr}_{\Y}(\tilde{X}) = \I_{\X},
    \end{aligned}
\end{equation}
where $\mathcal{U}=\I_{\Y}\otimes \mathcal{I}_{\X\to\Z}$.

For simplicity, one may equivalently define the variable  
\[ {X} = \left( \I_{\Y}\otimes \mathcal{I}_{\X\to\mathcal{Z}} \right)\tilde{X}\left( \I_{\Y}\otimes \mathcal{I}_{\X\to\Z} \right)^\dag \in  L(\Y \otimes \Z).\]
Then the  problem~\eqref{eq:opt_rigorous}  simplifies to
\begin{equation}
    \begin{aligned}\label{opt:main}
        \min_{X \in     L(\Y \otimes \Z)} \quad & \left\| \sigma - \frac{1}{2} {X} \right\|_* \\
        \text{s.t.} \quad &   {X}\succeq 0\\
                          & \operatorname{Tr}_{\Y}({X}) = \I_{\Z}.
    \end{aligned}
\end{equation}
In the remainder of this work, we no longer distinguish between $X$ and $\tilde{X}$. 

\subsection{Reduction via Group Invariance}
Group invariance is commonly utilized to simplify optimization problems. In this subsection, we use this property to further simplify Problem~\eqref{opt:main}. When $\sigma$ is a diagonal matrix, we have the following result.
		\begin{proposition}\label{prop:group-inv}
		If $\sigma$ is a diagonal matrix, then problem~\eqref{opt:main} admits a diagonal optimal solution.
	\end{proposition}
	\begin{proof}
		This result follows from group invariance. Let  $G$ be the subgroup of diagonal unitary matrices defined as
        \begin{equation}
            \begin{aligned}
                G = \left\{ \sum_{i,j=0}^1u_{i,j}\ket{j}_\Y\bra{j} \otimes \ket{i}_\Z\bra{i} \in L(\Y)\otimes L(\Z) ;\right.\\ \left. u_{ij}\in \{-1,1\},u_{00}u_{10}=u_{01}u_{11}   \right\}.
            \end{aligned}
        \end{equation}
		Define the group action of an element \(U \in G\) on $L(\Y\otimes \Z)$ as \(U(X) = U X U^\dag\). One can verify that 
		\begin{itemize}
			\item If $X \succeq 0$, then $U(X) \succeq 0$.
			\item If $\operatorname{Tr}_\Y(X)=\I_\Z$, then  $\operatorname{Tr}_\Y(U(X))=\I_\Z$.
			\item $U(\sigma)=\sigma$ and thus 
            \[\red{\left \|\sigma-\frac{1}{2}U(X) \right \|_*= \left \|\sigma-\frac{1}{2}X  \right \|_*}.\]
		\end{itemize}
		That is, if $X$ is a feasible solution to problem \eqref{opt:main}, then $U(X)$ is also a feasible solution, and the values of the objective function are equal. Thus, if $X$ is an optimal solution to the problem \eqref{opt:main}, then the matrix
		\[X_G\coloneq\frac{1}{|G|}\sum_{U \in G} U(X)\]
		is also an optimal solution, and it is invariant under group action, that is, $U(X_G) = X_G$ for all $U \in G$.

		For  $i,j,k,l\in \{0,1\}$ with $\ket{i}_\Y\ket{j}_\Z \neq \ket{k}_\Y\ket{l}_\Z  $ there exists an element 
		$U \in G$ such that 
		$$U \ket{i}_\Y\ket{j}_\Z=\ket{i}_\Y\ket{j}_\Z \quad $$
        and
        $$\quad  U \ket{k}_\Y\ket{l}_\Z=-\ket{k}_\Y\ket{l}_\Z,$$
		then         
        \begin{equation*}
         \begin{aligned}
            &\operatorname{Tr}( X_G \ket{k}_\Y  \bra{i} \otimes  \red{\ket{l}_{\Z}\bra{j}} )\\&=\operatorname{Tr}( UX_GU^\dag  \ket{k}_\Y  \bra{i} \otimes   \red{\ket{l}_{\Z}\bra{j}})\\
            &=\operatorname{Tr}( X_GU^\dag \ket{k}_\Y  \bra{i}\otimes   \red{\ket{l}_{\Z}\bra{j}} U )\\
            &=-\operatorname{Tr}( X_G  \ket{k}_\Y  \bra{i} \otimes   \red{\ket{l}_{\Z}\bra{j}} ),
            \end{aligned}
        \end{equation*}
		which implies that $X_G$ is a diagonal matrix.
	\end{proof}
	This result shows that when $\sigma$ is a diagonal matrix, constraining the feasible set to diagonal matrices reduces the original nuclear norm minimization problem~\eqref{opt:main} to an equivalent   $\ell_1$-norm minimization problem on vectors. 
	\begin{remark}
	    The existence of a diagonal optimal solution proved in Proposition~\ref{prop:group-inv} can be extended to any convex, $G$-invariant objective function $f: L(\mathcal{Y} \otimes \mathcal{Z}) \to \mathbb{R}$ i.e., $f(UXU^\dag) = f(X)$ for any $X$ and $U\in G$.
	\end{remark}
	 Note that, for any diagonal matrix $X$, the spectral calculus reduces to applying the function directly to the diagonal entries. Specifically, if $X=\operatorname{diag}(x_{11},x_{22},\dots,x_{nn})$, then $\operatorname{sign}(X)=\operatorname{diag}(\operatorname{sign}(x_{11}),\operatorname{sign}(x_{22}),\dots,\operatorname{sign}(x_{nn}))$.  
     
     \section{Counterexample}	
     Using the symmetry reduction established in Section~\ref{sec:2}, we significantly narrow the search space for a counterexample. In this section, we  will construct a pair of diagonal matrices $\sigma$ and $X$ to disprove Conjecture~\ref{conj:main}. The main result is as follows:
	\begin{theorem}
	There exists a 2-dimensional  optimal state transformation problem where \red{an optimal solution} $\Phi$ violates the conjectured optimality conditions in Conjecture~\ref{conj:main}. Moreover, it satisfies that
	$$ \I_{\Y}\otimes \operatorname{Tr}_{\Y}(HJ(\Phi)) \not\succeq H, $$
    and
    $$\quad H \not\succeq \I_{\Y}\otimes \operatorname{Tr}_{\Y}(HJ(\Phi)). $$
		\end{theorem}
        
     Now we are ready to construct the counterexample, the  counterexample is constructed via numerical methods\footnote{
      The source code for constructing and verifying the counterexample is available at \url{https://github.com/jty-AMSS/Quantum-Channel-Optimality-Conjecture}. AI tools were used to assist in the development and refinement of these scripts.}, \red{but the verification of optimality is analytic}.
	
	\begin{lemma}\label{lem:con-example}
	
	Let $\rho_0$ be defined as in \eqref{equ:def-rho},    define $\sigma$ by 
	\[
	\begin{aligned}
        \red{\sigma\coloneqq{}}&
        \red{\frac{11}{20}}\ket{0}_\Y \bra{0}\otimes \ket{0}_\Z\bra{0}+
        \red{\frac{1}{5}}\ket{0}_\Y \bra{0} \otimes \ket{1}_\Z \bra{1}+\\
		&\red{\frac{3}{20}}\ket{1}_\Y\bra{1}\otimes  \ket{0}_\Z \bra{0} 
        +\red{\frac{1}{10}}\ket{1}_\Y \bra{1} \otimes \ket{1}_\Z \bra{1},
	\end{aligned}
	\]
	and $X^*$ by
    	\[
	\begin{aligned}
		\red{X^*\coloneqq{}}&
        1\ket{0}_\Y\bra{0} \otimes  \ket{0}_\Z \bra{0}
        +\red{\frac{2}{5}}\ket{0}_\Y  \bra{0} \otimes \ket{1}_\Z \bra{1} +\\
		& \red{\frac{3}{5}}\ket{1}_\Y  \bra{1} \otimes \ket{1}_\Z \bra{1}.
	\end{aligned}
	\]
	Then $X^*$ is an optimal solution to the problem~\eqref{opt:main}. Consequently, the corresponding \red{quantum channel} $\Phi$ is also an optimal solution to the original problem \eqref{opt:org}. However, $\Phi$ does not satisfy the optimality conditions in \eqref{equ:Conj}.
\end{lemma}
\begin{proof}
	Since $\operatorname{Tr}_\Y(X^*)=\I_\Z$, $X^*$ is a feasible solution. 
    \red{By Proposition~\ref{prop:group-inv}, Problem~\eqref{opt:main} admits a diagonal optimal solution, so it suffices to verify the optimality of $X^*$ over the set of diagonal feasible solutions.} One can directly verify the  optimality of $X^*$, by parameterizing the feasible set of \eqref{opt:main} with two real parameters $\alpha, \beta \in [0,1]$. Specifically, the diagonal entries of any feasible diagonal matrix can be written as 
\begin{equation*}
    X = \sum_{i=0}^1 \sum_{j=0}^1 x_{ij} \ket{i}_\Y\bra{i} \otimes \ket{j}_\Z\bra{j},
\end{equation*}
with $x_{00}=\alpha, x_{10}=1-\alpha$ and $x_{01}=\beta, x_{11}=1-\beta$. 

Consequently, the objective function reduces to a scalar function $f(\alpha, \beta)$ on $[0,1]\times[0,1]$:
\begin{equation}
\begin{aligned}
f(\alpha, \beta)=  &\left\| \sigma - \frac{1}{2}X \right\|_{\red{*}} \\
=& \left( \left| \red{\frac{11}{20}} - \frac{\alpha}{2} \right| + \left| \frac{\alpha}{2} - \red{\frac{7}{20}} \right| \right) + \\
&\left( \left| \red{\frac{1}{5}} - \frac{\beta}{2} \right| + \left| \frac{\beta}{2} - \red{\frac{2}{5}} \right| \right) \\
\geq& \left|\red{\frac{11}{20}} - \red{\frac{7}{20}}\right| +\left|\red{\frac{1}{5}} - \red{\frac{2}{5}}\right| = \red{\frac{2}{5}}.
\end{aligned}
\end{equation}
\red{The inequality is directly from the triangle inequality.} Since the constructed $X^*$ (with $\alpha=1$ and $\beta=\red{\frac{2}{5}}$) attains this lower bound exactly, it is optimal. 
    
	Define 
	\[
	\begin{aligned}
		\Delta \coloneqq \sigma-\frac{1}{2}X^*=&\red{\frac{1}{20}}
        \ket{0}_\Y \bra{0} \otimes \ket{0}_\Z \bra{0}
        +\\&\red{\frac{3}{20}}\ket{1}_\Y\bra{1} \otimes  \ket{0}_\Z \bra{0}+ \\
		&\left (-\red{\frac{1}{5}}\right)\ket{1}_\Y\bra{1} \otimes \ket{1}_\Z \bra{1}.
	\end{aligned}
	\]
	Since $\Delta$ is a diagonal matrix, its spectral calculation is as follows:
	\[
	\begin{aligned}
		Y \coloneqq&	\operatorname{sign}(\Delta)\\
        =&\ket{0}_\Y \bra{0} \otimes
         \ket{0}_\Z \bra{0} 
         +\ket{1}_\Y\bra{1}\otimes \ket{0}_\Z \bra{0}+\\
		& (-1) \ket{1}_\Y\bra{1} \otimes \ket{1}_\Z\bra{1}.
	\end{aligned}
	\]
	Since $(\id_{L(\Y)}\otimes \Psi_{\rho_0})(\red{X})=\frac{1}{2}\red{X}$, its adjoint operator is $(\id_{L(\Y)}\otimes \Psi_{\rho_0}^*)(Y)=\frac{1}{2}Y$. As proposed in~\eqref{equ:def:H}, let   
    $$\red{H\coloneq (\id_{L(\Y)}\otimes \Psi_{\rho_0}^*)(Y)},$$
    \red{by the preceding discussion, $H=\frac{1}{2}Y$, therefore}
	\begin{equation}\label{H:compute}
		 \red{\begin{aligned}
       H=&\frac{1}{2}\ket{0}_\Y \bra{0} \otimes\ket{0}_\Z \bra{0}
            +\frac{1}{2}\ket{1}_\Y\bra{1}\otimes  \ket{0}_\Z\bra{0}+\\
            & \left(-\frac{1}{2}\right)\ket{1}_\Y\bra{1} \otimes \ket{1}_\Z \bra{1},
		\end{aligned}}
	\end{equation}
	and 
    \[
\begin{aligned}
HX^* =& \red{\frac{1}{2}}\ket{0}_\Y\bra{0}\otimes\ket{0}_\Z\bra{0}+ \\
    &\left(-\red{\frac{3}{10}}\right)\ket{1}_\Y\bra{1}\otimes\ket{1}_\Z\bra{1},
\end{aligned}
\]
	then 
	\[\begin{aligned}
		\operatorname{Tr}_{\Y}(HX^*)= \red{\frac{1}{2}} \ket{0}_\Z\bra{0} -\red{\frac{3}{10}} \ket{1}_\Z\bra{1},
	\end{aligned}
	\]
	and thus 
	\begin{equation}\label{1Y-tensor-H:compute}
		\begin{aligned}
			\I_{\Y}\otimes \operatorname{Tr}_{\Y}(HX^*) =&\red{\frac{1}{2}}\ket{0}_\Y \bra{0}\otimes \ket{0}_\Z \bra{0} +\\
            &\left(-\red{\frac{3}{10}}\right)\ket{0}_\Y \bra{0} \otimes\ket{1}_\Z\bra{1}+\\
			&\red{\frac{1}{2}}\ket{1}_\Y\bra{1}\otimes \ket{0}_\Z \bra{0}+\\
            &\left(-\red{\frac{3}{10}}\right) \ket{1}_\Y\bra{1} \otimes \ket{1}_\Z\bra{1}.
		\end{aligned}
	\end{equation}
By comparing the coefficients of terms $\ket{1}_\Y \bra{1}\otimes \ket{1}_\Z \bra{1}$ and $\ket{0}_\Y\bra{0}\otimes \ket{1}_\Z \bra{1}$ in \eqref{H:compute} and \eqref{1Y-tensor-H:compute}, we obtain that
	$$ \I_{\Y}\otimes \operatorname{Tr}_{\Y}(HX^*) \not\succeq H \quad \text{and} \quad H \not\succeq \I_{\Y}\otimes \operatorname{Tr}_{\Y}(HX^*). $$
   
	\end{proof}
		\section{Conclusion}
        In this paper, we study the optimality conditions for \red{the} trace distance based
optimal state transformation problem.  While the prior research confirmed dual certificate is uniquely determined via spectral calculus when the error matrix is full-rank, the efficiently verifiable optimality conditions for  the rank-deficient (singular) error matrices have remained an open question.
        
        By exploiting the problem's symmetry structures, we reduce the space of quantum channels  to a diagonal subspace and construct a counterexample disproving the conjecture in \cite{Coutts2021certifying}.  Our counterexample  shows  that the conjectured spectral condition is insufficient to certify the optimality \red{in the general case}. 
        
        It is worth noting that in \cite{Coutts2021certifying}, the authors reported that dual certificates computed via the spectral sign function are valid in their numerical experiments. Our results provide a  counterpoint by presenting an explicit case where this specific spectral calculus  fails. Consequently, when the error matrix is rank-deficient, the information obtained from spectral calculus alone is insufficient to certify optimality, as \red{a} dual certificate may exist elsewhere within the subdifferential set.

 \section*{Acknowledgments}
 We acknowledge Tianshi Yu and Lihong Zhi for their insightful discussions and verification of the results during the preparation of this manuscript.  
This research is supported by the Postdoctoral Fellowship Program of CPSF under Grant Number GZC20252039, the National Key R\&D Program of China (2023YFA1009401) and  the Basic Science Center Program (No: 12288201) of the National Natural Science Foundation of China.

\section*{Author contribution statement}\label{sec:author-contribution}
Jianting Yang conceived the study, worked out the theoretical results, completed all the proofs, and wrote the manuscript. Large language models were used only for numerical scripts and language polishing.

\bibliographystyle{quantum}
\bibliography{opt}

\end{document}